\documentclass[submission, copyright, creativecommons]{eptcs}

\usepackage{breakurl}             
\usepackage{underscore}           
\usepackage{amsfonts,url,graphicx}
\usepackage{color,soul} 
\usepackage{listings}
\usepackage{amsthm}
\usepackage{xspace}

\providecommand{\event}{Proceedings of the Workshop on Coalgebra,\\Horn Clause Logic Programming and Types (CoALP-Ty'16)\\
Edinburgh, Scotland, UK, November 28-29, 2016}

\title{Extending Coinductive Logic Programming with Co-Facts}
\author{Davide Ancona \qquad\qquad Francesco Dagnino \qquad\qquad Elena Zucca
\institute{DIBRIS, University of Genova}
\email{\{davide.ancona,elena.zucca\}@unige.it, fra.dagn@gmail.com}
}
\def\titlerunning{Extending Coinductive Logic Programming with Co-Facts}
\def\authorrunning{Ancona, Dagnino \& Zucca}

\newcommand{\emptygoal}{\varepsilon} 

\newcommand{\dom}{\mbox{\sf dom}}

\newcommand{\goalsep}{\,\Box\,}

\newcommand{\Rule}[4]{\rn{#1}{\displaystyle\frac{#2}{#3}}\ {\scriptsize{#4}}}
\newcommand{\RuleNoName}[3]{{\displaystyle\frac{#1}{#2}}\ {\scriptsize{#3}}} 

\newcommand{\clause}[2]{#1 \mathrel{:\!\!-} #2}

\newcommand{\progcoax}[2]{{\langle}{#1},{#2}{\rangle}}
\newcommand{\cosemSep}[1]{\vdash_{#1}}
\newcommand{\semSep}[1]{\vdash_{#1}}
\newcommand{\cosem}[5]{#1 \cosemSep{#2} \langle {#3} \goalsep {#4} \rangle \Rightarrow #5}
\newcommand{\sem}[4]{\semSep{#1} \langle {#2} \goalsep {#3} \rangle \Rightarrow #4}
\newcommand{\cosemMain}[3]{\cosemSep{#1} {#2} \Rightarrow {#3} }
\newcommand{\rn}[1]{{\scriptsize \textrm{({#1})}}}

\newcommand{\UnifyEq}[2]{E^{{#1}={#2}}}
\newcommand{\Gsol}[1]{\textsf{gsol}({#1})}

\newcommand{\Var}[1]{\textsf{V}({#1})}

\newcommand{\CompSubst}[2]{{#1}{#2}}
\newcommand{\Inv}[1]{{#1}^{-1}}

\newtheorem{theorem}{Theorem}[section]
\newtheorem{lemma}{Lemma}[section]
\newtheorem{corollary}{Corollary}[section]

\newcommand{\refToFigure}[1]{Fig.~\ref{fig:#1}}
\newcommand{\refToSection}[1]{Sect.~\ref{sect:#1}}
\newcommand{\refToTheorem}[1]{Theorem~\ref{theo:#1}}
\newcommand{\refToCorollary}[1]{Corollary~\ref{cor:#1}}
\newcommand{\refToLemma}[1]{Lemma~\ref{lemma:#1}}

\newcommand{\refToApp}[1]{Appendix~\ref{sect:#1}}
\newcommand{\Space}{\hskip 0.7em}

\newcommand{\fun}[3]{#1:#2\rightarrow#3}

\newcommand{\ApplySubst}[2]{#1#2}
\newcommand{\HU}{\textsf{HU}}
\newcommand{\HB}{\textsf{HB}}
\newcommand{\coHU}{\textsf{co-}\HU}
\newcommand{\coHB}{\textsf{co-}\HB}
\newcommand{\Ind}[1]{\textit{Ind}(#1)}
\newcommand{\CoInd}[1]{\textit{CoInd}(#1)}
\newcommand{\Generated}[2]{\textit{Gen}(#1,#2)}
\newcommand{\prog}{\mathit{P}}
\newcommand{\Ground}[1]{\textsf{ground}(#1)}
\newcommand{\coGround}[1]{\textsf{co-ground}(#1)}
\newcommand{\cofacts}{\mathit{C}}

\newcommand{\Extended}[2]{{#1_{{\sqcup}#2}}}
\newcommand{\cohyp}{\mathit{S}}

\newcommand{\Goal}{\mathit{G}}

\newcommand{\Op}[1]{{\textit{T}_{#1}}}
\newcommand{\GoalMax}{\Goal^{\texttt{\footnotesize max}}}
\newcommand{\OpSem}[2]{\mathit{OpSem}(#1,#2)}
\newcommand{\Agree}[2]{#1\vert#2}

\newcommand{\lstmath}[1]{\mbox{\lstinline{#1}}}
\newcommand{\der}{\nabla}
\newcommand{\blindfrac}[1]{\begin{array}{c}\\ {#1}\end{array}}

\newif\ifsubmit
\submittrue
\ifsubmit
\newcommand{\EZ}[1]{{#1}} 
\newcommand{\FD}[1]{{#1}} 
\newcommand{\DA}[1]{{#1}} 
\newcommand{\EZComm}[1]{} 
\newcommand{\FDComm}[1]{} 
\newcommand{\DAComm}[1]{} 
\else
\newcommand{\EZ}[1]{\textcolor{blue}{#1}} 
\newcommand{\FD}[1]{\textcolor{red}{#1}} 
\newcommand{\DA}[1]{\textcolor{magenta}{#1}} 
\newcommand{\EZComm}[1]{{\scriptsize\textcolor{blue}{[\bf{Elena: }#1}]}}
\newcommand{\FDComm}[1]{{\scriptsize\textcolor{red}{[\bf{Francesco: }#1}]}}
\newcommand{\DAComm}[1]{{\scriptsize\textcolor{magenta}{[\bf{Davide: }#1}]}}
\fi

\begin{document}
\maketitle

\begin{abstract}
We introduce a generalized logic programming paradigm where programs, consisting of facts and rules with the usual syntax, can be enriched by \emph{co-facts}, which syntactically resemble facts but have a special meaning. As in coinductive logic programming, interpretations are subsets of the complete Herbrand basis, including infinite terms. However, the intended meaning (declarative semantics) of a program is a fixed point which is not necessarily the least, nor the greatest one, but is determined by co-facts. In this way, it is possible to express predicates on non well-founded structures, such as infinite lists and graphs, for which the coinductive interpretation would be not precise enough. Moreover, this paradigm nicely subsumes standard (inductive) and coinductive  logic programming, since both can be expressed by a particular choice of co-facts, hence inductive and coinductive predicates can coexist in the same program. We illustrate the paradigm by examples, and provide declarative and operational semantics, proving the correctness of the latter. Finally, we describe a prototype meta-interpreter.
\end{abstract}

\section{Introduction} \label{sect:intro}
Coinductive logic programming \cite{Simon06,SimonEtAl06,SimonEtAl07,AnconaDovier15} extends standard logic programming with the ability of reasoning
about infinite objects and their properties. Whereas syntax of logic programs remains the same, semantics is different. To illustrate this, let us consider the following logic program which defines some predicates on  lists of integers, constructed with the standard function symbols \lstinline{[]} of arity 0 for the empty list and \lstinline{[$\_$|$\_$]} of arity 2 for the list consisting of a head element and a tail. For simplicity, we will consider built-in integers, as they are in Prolog. 

\begin{lstlisting}
all_pos([]).
all_pos([N|L]) :- N>0, all_pos(L).

member(X,[X|_]).
member(X,[Y|L]):-X\=Y, member(X,L).

max([N],N).
max([N|L],M2) :- max(L,M), M2 is max(N,M).
\end{lstlisting}

The expected meaning is that \lstinline{all_pos($l$)} holds if all the elements of $l$ are positive, \lstinline{member($x$,$l$)} if $x$ is an element of $l$,   \lstinline{max($l$,$n$)} if $n$ is the greatest element of list $l$.
As will be illustrated in detail in \refToSection{cofacts}, in standard logic programming terms are inductively defined, that is, are finite, and predicates are inductively defined as well. In the example program, only finite lists are considered, such as, e.g., \lstinline{[1|[2|[]]]}, and the three predicates are correctly defined on such lists. 

In  coinductive logic programming, instead,  terms are coinductively defined, that is, can be infinite, and predicates are coinductively defined as well. In the example program, also infinite lists such as \lstinline{[1|[2|[3|[4|...]]]]}, are considered, and the coinductive interpretation of the predicate \EZ{\lstinline{all_pos}} gives the expected meaning on such lists. However, this is not the case for the other two predicates: for \lstinline{member} the correct interpretation is the inductive one\EZComm{Davide avevi fatto da qualche parte un bel discorso sulle propriet\`a di safety e liveness, pensi ci stia bene?},  whereas for \lstinline{max} neither the inductive nor the coinductive interpretation are correct: with the former the predicate is always false on infinite lists, with the latter \lstinline{max($l$,$n$)} is true whenever $n$ is greater than all the elements of $l$. 

The last example shows that the coinductive interpretation of predicates is sometimes \emph{not precise enough}, in the sense that also wrong facts are included. This problem can be found, and has been studied, also in other programming paradigms, and some solutions have been proposed which allow the programmer
to interpret corecursive definitions not in the standard coinductive way \cite{JeanninEtAl12,JeanninEtAl13,Ancona13,AnconaZucca12,AnconaZucca13}.

In this paper, we solve the problem in a more foundational way, by applying to the case of logic programs a notion recently introduced in the more general framework of inference systems \cite{AnconaEtAl17} (indeed, a logic program can be seen as an inference system where judgments are atoms). That is, programs, consisting of facts and rules with the usual syntax, can be enriched by \emph{co-facts} (corresponding to \emph{coaxioms} in \cite{AnconaEtAl17}), which syntactically resemble facts but have a special meaning: intuitively, they can only be applied ``at infinite depth'' in a proof tree, as will be formally defined in the following. By adding co-facts, the intended meaning (declarative semantics) of a program can be a fixed point which is not necessarily the least, nor the greatest one. 

In this way, it is possible to express predicates on non well-founded structures, such as infinite lists and graphs, for which the coinductive interpretation would be not precise enough. Moreover, this paradigm nicely subsumes standard (inductive) and coinductive  logic programming, since both can be expressed by a particular choice of co-facts, hence inductive and coinductive predicates can coexist in the same program. 

For what concerns operational semantics, in coinductive logic programming standard SLD resolution is replaced by co-SLD resolution \cite{SimonEtAl06,AnconaDovier15}, which, roughly speaking, keeps trace of the already encountered goals, called \emph{(dynamic)  coinductive hypotheses}, so that, when a goal is encountered the second time, it is considered successful. In this paper, we define an operational semantics of logic programs which is similar to co-SLD resolution, but takes co-facts into account, and prove its correctness with respect to the proposed declarative semantics. The proof is interesting since it is a \EZ{non-trivial} application of the \emph{bounded coinduction principle} introduced in \cite{AnconaEtAl17}. 

The operational semantics we define is in big-step style, as also proposed for co-SLD resolution \cite{AnconaDovier15}, and, hence, is amenable
for directly deriving an implementation; indeed, we have implemented in SWI-Prolog a prototype meta-interpreter, whose clauses are driven by
our operational semantics, and with which we have been able to successfully test all examples shown in this paper, and many others.  

The rest of the paper is organized as follows: in \refToSection{cofacts} we introduce logic programs with co-facts, their declarative semantics, and the bounded coinduction principle to reason on such programs, 
illustrating the notions with some examples. In \refToSection{big-step} we formally define operational semantics, show a derivation example, and prove soundness with respect to declarative semantics. In \refToSection{impl} we describe the prototype meta-interpreter. In  \refToSection{conclu} we summarize our contribution, survey related work, and discuss further work.

\section{Co-facts}\label{sect:cofacts}
We recall some notions about standard \cite{Lloyd87,Apt97} and coinductive \cite{Simon06,SimonEtAl06,SimonEtAl07} logic programming.

Assume  a \emph{signature} consisting of sets of \emph{predicate symbols} $p$, \emph{function symbols} $f$, and \emph{variable symbols} $X$, each one with an associated \emph{arity} $\geq 0$, being $0$ for variable symbols. A function of arity $0$ is called a \emph{constant}. 

\emph{Terms} are (possibly infinite) trees where nodes are labeled with function and variable symbols and the number of children of a node corresponds to the symbol arity (for a more formal definition based on paths see, e.g., \cite{AnconaDovier15}). \emph{Atoms} are (possibly infinite) trees where the root is labeled with a predicate symbol and other nodes are labeled with function and variable symbols, also accordingly with the arity. Terms and atoms are \emph{ground} if they do not contain variable symbols, \emph{finite} (or \emph{syntactic}) if they are finite trees.


A \emph{logic program} $\prog$ is a set of \emph{(definite) clauses} of shape $\clause{A}{B_1,\dots,B_n}$, where $A, B_1, \ldots, B_n$ are finite atoms. A clause where $n=0$ is called a \emph{fact}, otherwise it is called a \emph{rule}.

A \emph{substitution} $\theta$ is a mapping from a finite subset of variables into terms. We write $\ApplySubst{t}{\theta}$ for the application of a substitution $\theta$ to a term $t$. We call $\ApplySubst{t}{\theta}$ an \emph{instance} of $t$. These notions can be analogously defined on atoms and clauses. A substitution is \emph{ground} (or \emph{grounding}) \EZ{if} it maps variables into ground terms, it is \emph{syntactic} if it maps variables into finite (syntactic) terms. 
The \emph{declarative semantics} of a logic program describes its meaning in an abstract way, as the set of ground atoms which are defined to be true by the program, in a sense to be made precise depending on the kind of declarative semantics we choose.
In the following \EZ{paragraphs,} we briefly recall the standard declarative semantics of logic programs, then their coinductive declarative semantics, and finally we define the declarative semantics \emph{generated by co-facts} and discuss its advantages. 

\paragraph{Standard declarative semantics} In the standard declarative semantics of logic programs, only finite terms and atoms are considered. 
The \emph{Herbrand universe} $\HU$ is defined as the set of finite ground terms, and the \emph{Herbrand
base} $\HB$ as the set of finite ground atoms. Sets  $I \subseteq \HB$ are called \emph{interpretations}. 
Given a logic program $\prog$, the (one step) inference operator $\fun{\Op{\prog}}{\wp(\HB)}{\wp(\HB)}$ is defined as follows:

$$\Op{\prog}(I)  = \{A \mid (\clause{A}{B_1,\dots,B_n}) \in \Ground{\prog}, \{B_1,\dots,B_n\} \subseteq I\}$$

\noindent where $\Ground{\prog}$ is the set of instances of clauses in $\prog$ obtained by a ground syntactic substitution.

An interpretation is a \emph{model} of a program $\prog$ (is \emph{closed} with respect to $\prog$) if $\Op{\prog}(I) \subseteq I$. The standard declarative semantics of $\prog$ is the least interpretation which is a model taking as order set inclusion, that is, the intersection of all closed interpretations. Defining a \emph{proof tree} for a ground atom $A$ as a tree where the root is $A$, nodes are ground instances of rules, and leaves are ground instances of facts, the standard declarative semantics can be equivalently characterized as the set of finite ground atoms which have a finite proof tree. 


It is easy to see that, with this definition, the predicates of the example program introduced in \refToSection{intro} have the expected meaning on finite lists. For instance, for the predicate \lstinline{max} we obtain all atoms \lstinline{max($l$,$n$)} where $l$ is a (ground term representing a) finite list and $n$ is the greatest element of $l$.

 \paragraph{Co-inductive declarative semantics}
A limit of the standard declarative semantics described above is that we cannot define predicates on non-well-founded structures, such as infinite lists or graphs.  
Considering our running example, we would like to define the predicates \lstinline{all_pos}, \lstinline{member}, and \lstinline{max} on infinite lists as well. 
To obtain this, first of all infinite terms and atoms should be included. 
The \emph{complete Herbrand Universe} $\coHU$~\cite{Lloyd87}
is the set of (finite and infinite) ground terms.  
The \emph{complete Herbrand base} $\coHB$ is the set of (finite and infinite) ground atoms.
Sets  $I \subseteq \coHB$ are called
\emph{co-interpretations}.  

We can define the (one step) inference operator $\fun{\Op{\prog}}{\wp(\coHB)}{\wp(\coHB)}$ analogously to that above:
$$\Op{\prog}(I)  = \{A \mid (\clause{A}{B_1,\dots,B_n}) \in \coGround{\prog}, \{B_1,\dots,B_n\} \subseteq I\}$$
where $\coGround{\prog}$ is the set of instances of clauses in $\prog$ obtained by a ground substitution.\footnote{\EZ{An instance of a clause in $\prog$ obtained by mapping some variables into infinite ground terms belongs to $\coGround{\prog}$, but does not belong to $\Ground{\prog}$.}}

Again as above, a co-interpretation is a \emph{model} of a program $\prog$ (is \emph{closed} with respect to $\prog$) if $\Op{\prog}(I) \subseteq I$, and we can consider the the least co-interpretation which is a model, that is, the intersection of all closed co-interpretations. This co-interpretation, called in the following \emph{inductive (declarative) semantics} of $\prog$ and denoted $\Ind{\prog}$, generalizes the standard declarative semantics of predicates to infinite terms. However, in order to prove predicates on such infinite terms, we would like to allow infinite proof trees as well, and to this end a different declarative semantics can be adopted, explained below. 

A co-interpretation \EZ{$I$} is a \emph{co-model} of $\prog$ (is \emph{consistent} with respect to $\prog$) if and only
if $I \subseteq \Op{\prog}(I)$. The \emph{coinductive (declarative) semantics} of $\prog$, denoted $\CoInd{\prog}$, is the greatest co-interpretation which is a co-model taking as order set inclusion, that is, the union of all consistent co-interpretations. Equivalently, $\CoInd{\prog}$ can be  characterized as the set of ground atoms which have a (finite or infinite) proof tree. 

Let us analyze what happens in the running example. Consider first the predicate \lstinline{all_pos}. It is easy to see that, with the inductive semantics $\Ind{\prog}$, this predicate has the expected meaning only on finite lists (as in the standard semantics) and on infinite lists which have a non positive element. However, for $l$ term representing an infinite list of positives, e.g., $l=$\lstinline{[1|[2|[3|[4|...]]]]}, the atom \lstinline{all_pos($l$)} cannot be proved by a finite proof tree, hence the predicate turns out to be false. With the coinductive semantics, instead, we can prove \lstinline{all_pos($l$)} by the infinite proof tree shown below:

\begin{small}
$$\RuleNoName{\RuleNoName{\RuleNoName{\vdots}{\texttt{all_pos([3|[4|...]])}}{}}{\texttt{all_pos([2|[3|[4|...]]])}}{}}{\texttt{all_pos([1|[2|[3|[4|...]]]])}}{}$$
\end{small}

Hence, this predicate is a typical example where coinductive semantics is necessary (when including infinite terms), and provides the expected meaning.
However, this is not the case for the other two predicates. More precisely, for $l$ (term representing an) infinite list:
\begin{itemize}
\item \lstinline{member($x$,$l$)} always holds. In this case, the coinductive semantics $\CoInd{\prog}$ seems not to be the right choice, and the desired semantics is obtained by taking the inductive semantics (least model) on infinite lists as well, that is, $\Ind{\prog}$. For this reason, in coinductive logic programming the programmer can specify by a special notation that some predicates should be interpreted inductively rather than coinductively. 
\item \lstinline{max($l$,$n$)} holds whenever $n$ is greater than all the elements of $l$. In this case, the coinductive semantics \emph{includes} the desired semantics (and it is necessary for this, since \lstinline{max($l$,$n$)} would never hold in the inductive semantics), but it is not precise enough. \EZ{Indeed, for an infinite (regular) list $l=$\texttt{[1|[2|[1|[2|...]]]]}, we can prove \texttt{max($l$,2)} as expected, but we can also prove, e.g., \texttt{max($l$,4)}, as shown by the infinite proof trees (T1) and (T2), respectively, shown below:}\label{T1-T2}
\EZ{\begin{small}
$$\Rule{T1}{\RuleNoName{\RuleNoName{\vdots}{\texttt{max([1|[2|[1|...]]],2)}}{}}{\texttt{max([2|[1|[2|...]]],2)}}{}}{\texttt{max([1|[2|[1|...]]],2)}}{}
\Space
\Rule{T2}{\RuleNoName{\RuleNoName{\vdots}{\texttt{max([1|[2|[1|...]]],4)}}{}}{\texttt{max([2|[1|[2|...]]],4)}}{}}{\texttt{max([1|[2|[1|...]]],4)}}{}$$
\end{small}}
\end{itemize}

 \paragraph{Declarative semantics generated by co-facts}
 In order to define this semantics, first of all the syntax is slightly generalized allowing, besides facts and rules, \emph{co-facts}, written \lstinline{.A}.
 Co-facts are finite atoms, hence syntactically resemble facts, but have a special meaning.\EZComm{dire qui o nelle conclu che si generalizza in modo ovvio alle regole, in effetti pensandoci non si vede il perch\'e della restrizione - anche per gli inference system}
\EZComm{scrivo per ora i co-fatti mettendoci un punto prima invece che dopo}
Below is the version equipped with co-facts of the running example.

\begin{lstlisting}
all_pos([]).
all_pos([N|L]) :- N>0, all_pos(L).
.all_pos(_)

member(X,[X|_]).
member(X,[Y|L]) :- X\=Y, member(X,L).

max([N],N).
max([N|L],M2) :- max(L,M), M2 is max(N,M).
.max([N|_],N)
\end{lstlisting}
\noindent In the following, the metavariable $\cofacts$ denotes a set of co-facts (finite atoms).

 The \emph{(declarative) semantics} of a program $\prog$ \emph{generated by co-facts $\cofacts$}, denoted $\Generated{\prog}{\cofacts}$, is defined as follows.
\begin{itemize}
\item First, we consider the program ${\Extended{\prog}{\cofacts}}$ \EZComm{forse metterci proprio l'unione?} obtained by enriching $\prog$ by the co-facts in $\cofacts$ considered as facts, and we take its inductive semantics $\Ind{\Extended{\prog}{\cofacts}}$.
\item Then, we take the largest co-model which is included in $\Ind{\Extended{\prog}{\cofacts}}$. \EZ{In other words, we take the coinductive interpretation of $\prog$ where, however, clauses are instantiated only on elements of $\Ind{\Extended{\prog}{\cofacts}}$.}
\end{itemize}
\EZ{Note that $\Generated{\prog}{\cofacts}$ is different from $\CoInd{\prog}\cap\Ind{\Extended{\prog}{\cofacts}}$. For instance, let $\prog$ be the program}
\begin{lstlisting}
p(0) :- p(0), p(1)
p(1) :- p(0), p(1)
\end{lstlisting}
\EZ{and $\cofacts$ be the singleton set consisting of the co-fact}
\begin{lstlisting}
.(p(0))
\end{lstlisting}
\EZ{Then, $\CoInd{\prog}=\{\texttt{p(0)}, \texttt{p(1)}\}$, and $\Ind{\Extended{\prog}{\cofacts}}=\{\texttt{p(0)}\}$. Hence the intersection is $\{\texttt{p(0)}\}$, whereas $\Generated{\prog}{\cofacts}=\emptyset$.}

As we have shown in \cite{AnconaEtAl17} in the more general framework of inference systems, $\Generated{\prog}{\cofacts}$ corresponds to a fixed-point of the operator $\Op{\prog}$ which is neither the greatest, nor the least one.

In terms of proof trees, $\Generated{\prog}{\cofacts}$ is the set of ground atoms which have a (finite or infinite) proof tree in $\prog$ whose nodes all have a finite proof tree in ${\Extended{\prog}{\cofacts}}$. 
Taking this semantics, all the predicates in the running example get the expected meaning. Indeed, for $l$ (term representing an) infinite list:
\begin{itemize}
\item \lstinline{all_pos($l$)} holds for $l$  infinite list of positives, since the atom \lstinline{all_pos($l$)} has the previously shown infinite proof tree, and all its nodes have a (trivial) finite proof tree in ${\Extended{\prog}{\cofacts}}$, consisting in an instantiation of the co-fact.
\item \lstinline{member($x$,$l$)} only holds if $x$ belongs to $l$. Otherwise, there would exist an infinite proof tree in  $\prog$, but its nodes have 
no finite proof tree in  ${\Extended{\prog}{\cofacts}}$, since there are no co-facts, and the only fact is not applicable.
\item \lstinline{max($l$,$n$)} only holds when $n$ is the greatest element of $l$. In this case, as shown \EZ{at page \pageref{T1-T2}}, there exists an infinite proof tree in $\prog$ whenever 
$n$ is greater than all the elements of $l$. However, a finite proof tree in ${\Extended{\prog}{\cofacts}}$ for each node only exists when 
$n$ belongs to the list, hence is actually the greatest element. For instance, the two nodes of the infinite proof tree \EZ{(T1)} for \texttt{max([1|[2|[1|[2|...]]]],2)} have the finite proof trees \EZ{(FT1)} and \EZ{(FT2)} shown below:
\begin{small}
$$\Rule{\EZ{FT1}}{\RuleNoName{}{\texttt{max([2|[1|[2|...]]],2)}}{}}{\texttt{max([1|[2|[1...]]],2)}}{}\Space\Rule{\EZ{FT2}}{}{\texttt{max([2|[1|[2|...]]],2)}}{}$$
\end{small}
\end{itemize}
\EZ{whereas there is no finite proof tree for the nodes of the infinite proof tree (T2).}
 In other words, co-facts allow the programmer to ``filter out'' atoms which should not be true, making the semantics precise.

\EZComm{se il discorso sotto vi sembra prolisso possiamo toglierlo, il senso era di anticipare un po' quello che succede nella sem. operazionale}
  Note that the condition to have a finite proof tree in ${\Extended{\prog}{\cofacts}}$ trivially holds for nodes which are roots of a finite subtree (a finite proof tree in $\prog$ is a finite proof tree in ${\Extended{\prog}{\cofacts}}$ as well), hence is only significant for nodes which occur in an infinite path. Moreover, if the infinite tree is \emph{rational}, that is, has a finite number of different subtrees, then an infinite path always consists of (possibly) a finite prefix and a period. Hence, if the condition holds for the first node of the period, then it holds for all the other nodes. In the example \EZ{of (T1) above, another finite proof tree for the second node \texttt{max([2|[1|[2|...]]],2)} can be obtained from (FT1) for the first node:
\begin{small}
$$\RuleNoName{\RuleNoName{\RuleNoName{}{\texttt{max([2|[1|[2|...]]],2)}}{}}{\texttt{max([1|[2|[1|...]]],2)}}{}}{\texttt{max([2|[1|[2|...]]],2)}}{}$$
\end{small}}

Correspondingly, in the operational semantics which will be provided in \refToSection{big-step}, standard SLD resolution in ${\Extended{\prog}{\cofacts}}$ is triggered when an atom is encountered the second time.

\EZComm{forse aggiungere come esempio anche \lstinline{noRep}}

\EZComm{forse dare nomi a fatti regole e coassiomi per mostrarli nell'albero di prova?}

Note also that, as the examples above clearly show, the inductive and coinductive semantics can be obtained as special cases of the semantics generated by co-facts of a program, notably:
\begin{itemize}
\item the inductive semantics when the set of co-facts is empty;
\item the coinductive semantics when the set of (ground instances of) co-facts is $\coHB$.
\end{itemize} 
That is, co-facts allow to mix together, without any need of a special notation, predicates for which the appropriate interpretation is either inductive or coinductive, and to express predicates for which the appropriate interpretation is neither of the two, as the \lstinline{max} example shows.

Let $\Goal$ be a set of ground atoms, corresponding to the intended meaning of some predicates.

In order to prove that the the atoms in $\Goal$ are defined to be true by a program $\prog$ enriched by co-facts $\cofacts$, that is,  $\Goal\subseteq\Generated{\prog}{\cofacts}$, we can use the following \emph{bounded coinduction principle}, which is a generalization of the standard coinduction principle. This and other proof techniques are illustrated in the more general framework of  inference systems in \cite{AnconaEtAl17}. 

\begin{theorem}[Bounded coinduction principle]\label{theo:bcoind}If the following two conditions hold:
\begin{description}
\item[Boundedness] $\Goal \subseteq \Ind{\Extended{\prog}{\cofacts}}$, that is, each atom in $\Goal$ has a finite proof tree in ${\Extended{\prog}{\cofacts}}$
\item [Consistency] $\Goal \subseteq \Op{\prog}(\Goal)$, that is, for each atom $A\in \Goal$,\\
$\clause{A}{B_1,\dots,B_n} \in\coGround{\prog}$ for some $B_1,\ldots,B_n\in\Goal$
\end{description}
then $\Goal\subseteq\Generated{\prog}{\cofacts}$.
\end{theorem}
\begin{proof}
The two conditions corresponds to require that $\Goal$ is a co-model of $\prog$ (is \emph{consistent} with respect to $\prog$) which is \emph{bounded} by (included in) $\Ind{\Extended{\prog}{\cofacts}}$, and $\Generated{\prog}{\cofacts}$ is defined as the largest such co-model.
\end{proof}

The standard coinduction principle can be obtained as a specific instance of the principle above, when (ground instances of) $\cofacts$ coincide with $\coHB$; for this particular case the first condition trivially holds.

We illustrate the proof technique by formally proving that, in the running example, we can derive all atoms in the set $\GoalMax=\{{\texttt{max}(n,l)}\mid n\ \mbox{greatest element of}\ l\}$.
\begin{itemize}
\item To prove boundedness, we have to show that, for each \lstinline{max($l$,$n$)} such that $n$ \EZ{is the} greatest element of $l$, \lstinline{max($l$,$n$)} has a \emph{finite} proof tree in $\Extended{\prog}{\cofacts}$.
This can be easily shown. Indeed, if $n$ is the greatest element of $l$, then $l=[n_1| [ \ldots [n_k | l'] ] ]$ with $n_k=n$, $n_i\leq n$ for $i\in[1..k-1]$. Hence, \lstinline{max($l$,$n$)} has a finite proof tree in $\Extended{\prog}{\cofacts}$ which consists in $k-1$ nodes which are instances of the rule and a leaf which is an instance of the co-fact (a concrete example for \lstinline{max([2|[1|[2|1|...]]],2)} has been shown before). 
 \item To prove consistency, we have to show that, for each  \lstinline{max($l$,$n$)} such that $n$ greatest element of $l$, \lstinline{max($l$,$n$)} is the consequence of (an instance of) a clause with all atoms of the body in $\GoalMax$. This can be easily shown, indeed, if $n$ is the greatest element of $l$, then $l=[n' | l']$ with $n'\leq n$ and $n$ greatest element of $l'$. Hence, \lstinline{max($l$,$n$)} is the consequence of the following instance of the rule: \lstinline{max([$n'$|$l'$],$n$) :- max($l'$,$n$), $n$ is max($n'$,$n$)} where the atom \lstinline{max($l'$,$n$)} belongs to $\GoalMax$. 
\end{itemize}

The same proof technique will be used in the next section to show that the operational semantics is sound with respect to the semantics generated by co-facts.

\section{Big-step operational semantics}\label{sect:big-step}

In this section we define an operational counterpart of the semantics generated by co-facts introduced in \refToSection{cofacts}, and prove its correctness. \EZComm{dire che non pu\`o essere completa, inserire riferimenti su questo, magari mettere un esempio di questo?}

This operational semantics is a generalization of SLD \cite{Lloyd87,Apt97} and co-SLD \cite{Simon06,SimonEtAl06,SimonEtAl07} resolution. 
\EZComm{forse introdurre un nome, tipo co-facts driven SLD resolution?}
However, it is presented, rather than in the traditional small-step style, in big-step style, as introduced in \cite{AnconaDovier15}.
This style turns out to be simpler since coinductive hypotheses  (see below) can be kept local. Moreover, it naturally leads to an interpreter, and \EZ{the} proof of soundness with respect to declarative semantics is more direct since we compare two inference systems. 
For a proof of equivalence of big-step and small-step co-SLD resolution see \cite{AnconaDovier15}.

We introduce some notations. First of all, in this section the metavariable $A$ denotes finite (syntactic) atoms and the metavariables $s, t$ denote finite (syntactic) terms. 
An \emph{equation} has shape $s=t$ where $s$ and $t$ are finite terms. We write $E$ for a set of equations.\EZComm{non mi \`e chiaro se \`e veramente necessario parlare di solved form}
Equations allow a finite (syntactic) representation \FD{also} of \emph{rational} terms and atoms, that is, \FD{having} a finite number of different subterms, see, e.g., \cite{AdamekEtAl06, AnconaDovier15} for the details.
Finally, a \emph{goal} is a sequence of atoms, and the metavariable $G$ denotes a syntactic goal, that is, a sequence of finite atoms. The empty sequence is denoted by $\epsilon$. 

The operational semantics is defined by a judgment $\cosemMain{\progcoax{\prog}{\cofacts}} {G} {E}$, meaning that, given a program $\prog$ and a set of co-facts $\cofacts$, resolution of the (syntactic) goal $G$ succeeds \FD{and produces} a set of equations $E$ \FD{that} describes a solution of the goal.  For instance, assuming $\prog$ and $\cofacts$ from the previous section, which define, among others, the predicate \lstinline{max} computing the greatest element of a list, \FD{it would be possible to derive} the following judgment
\begin{quote}
$\cosemMain{\progcoax{\prog}{\cofacts}} {\texttt{max(L,2)}} {\{\texttt{L=[1|L1], L1=[2|L]]\}}}$
\end{quote}
meaning that a solution of the goal  \lstinline{max(L,2)} is the infinite term \lstinline{l = [1|[2|[1|2|  ... ]]]}. 

This judgment is defined following a schema which is similar to co-SLD resolution, in the sense that resolution \FD{keeps track} of the already encountered atoms, which are called \emph{(dynamic)  coinductive hypotheses} \cite{SimonEtAl06}. However, when the same atom $A$ is encountered the second time, rather than just considering $A$ successful as it happens in co-SLD resolution, standard SLD resolution of $A$ is triggered in the program ${\Extended{\prog}{\cofacts}}$ obtained by enriching $\prog$ by the co-facts in $\cofacts$.

Formally, two auxiliary judgments are introduced:
\begin{itemize}
\item $\cosem{\cohyp}{\progcoax{\prog}{\cofacts}}{G}{E}{E'}$, meaning that, given a program $\prog$ and a set of co-facts $\cofacts$, resolution of the goal represented by $G$ and $E$, under the \emph{coinductive hypotheses} $\cohyp$, succeeds producing a set of equations $E'$ which describes a solution of the goal.
\item $\sem{\prog}{G}{E}{E'}$ meaning that, given a program $\prog$, standard SLD resolution of the goal represented by $G$ and $E$ succeeds producing a set of equations $E'$ which describes a solution of the goal.
\end{itemize}

Rules (inductively) defining these judgments are shown in \refToFigure{opsem}.

\begin{figure}[t]
$
\begin{array}{l}
\Rule{empty}
{  }
{\sem{\prog}{\emptygoal}{E}{E}}
{}
\\[4ex]
\Rule{step}
{  
  \begin{array}{l}
  	\clause{p(t_1, \ldots, t_n)}{A_1, \ldots, A_m} 
      \mbox{ renaming of a clause in $P$ with 	fresh variables} \\
    E_1\cup \{s_1=t_1,\dots,s_n=t_n\} \mbox{ solvable} \\
    \sem{\prog}{A_1,\dots,A_m}{E_1\cup \{s_1=t_1,\dots,s_n=t_n\}}{E_2} \\
    \sem{\prog}{G_1,G_2}{E_2}{E_3} 
  \end{array}
}
{\sem{\prog}{G_1,p(s_1,\dots,s_n),G_2}{E_1}{E_3} }
{}
\\[4ex]
\Rule{co-empty}
{ }
{ \cosem{S}{\progcoax{\prog}{\cofacts}}{\emptygoal}{E}{E} }
{}
\\[4ex]
\Rule{co-step}
{  
  \begin{array}{l}
    \clause{p(t_1,\dots,t_n)}{A_1,\ldots,A_m} 
      \mbox{ renaming of a clause in $P$ with fresh variables} \\
    E_1\cup \{s_1=t_1,\dots,s_n=t_n\} \mbox{ solvable} \\
    \cosem{S\cup \{p(s_1,\dots,s_n)\}}{\progcoax{\prog}{\cofacts}}{A_1,\dots,A_m}{E_1\cup \{s_1=t_1,\dots,s_n=t_n\}}{E_2} \\
  \cosem{S}{\progcoax{\prog}{\cofacts}}{G_1,G_2}{E_2}{E_3} 
  \end{array}
}
{ \cosem{S}{\progcoax{\prog}{\cofacts}}{G_1,p(s_1,\dots,s_n),G_2}{E_1}{E_3} }
{}
\\[4ex]
\Rule{co-hyp}
{  
  \begin{array}{l}
    p(t_1,\dots,t_n) \in S \qquad
    E_1\cup \{s_1=t_1,\dots,s_n=t_n\} \mbox{ solvable} \\
    \sem{{\Extended{\prog}{\cofacts}}}{p(s_1,\dots,s_n)}{E_1\cup \{s_1=t_1,\dots,s_n=t_n\}}{E_2}\\
    \cosem{S}{\progcoax{\prog}{\cofacts}}{G_1,G_2}{E_2}{E_3}
  \end{array}
}
{ \cosem{S}{\progcoax{\prog}{\cofacts}}{G_1,p(s_1,\dots,s_n),G_2}{E_1}{E_3} }
{}
\\[4ex]
\Rule{main}
{  
  \cosem{\emptyset}{\progcoax{\prog}{\cofacts}}{G}{\emptyset}{E}
}
{ \cosemMain{\progcoax{\prog}{\cofacts}} {G} {E} }
{}
\end{array}
$
\caption{Big-step operational semantics}\label{fig:opsem}
\end{figure}

The first two rules define the judgment $\sem{\prog}{G}{E}{E'}$ corresponding to standard SLD resolution. In rule \rn{empty}, resolution of the empty goal succeeds. In rule \rn{step}, an atom $p(s_1,\dots,s_n)$ is selected from the goal to be resolved, and a clause of the program is chosen such that the selected atom unifies with the head of the clause, as expressed by the fact that adding the equations $s_1=t_1, \ldots,s_n=t_n$ to the current set we get a solvable set of equations. Then, resolution of the original goal succeeds if resolution of both the body of the clause and the remaining goal succeed.

The following three rules define the judgment $\cosem{\cohyp}{\progcoax{\prog}{\cofacts}}{G}{E}{E'}$ corresponding to an intermediate step of resolution. Rules \rn{co-empty} and \rn{co-step} are analogous to rules \rn{empty} and \rn{step}. The only difference is that, in rule \rn{co-step}, in the resolution of the body of the clause, the selected atom is added to the current set of coinductive hypotheses. 
In this way, rule \rn{co-hyp} can handle the case when resolution encounters the same atom for the second time, that is, more formally, the selected atom unifies with a coinductive hypothesis, as expressed by the fact that adding the equations $s_1=t_1, \ldots,s_n=t_n$ to the current set we get a solvable set of equations. In this case, standard SLD resolution of such atom is triggered in the program ${\Extended{\prog}{\cofacts}}$ obtained by enriching $\prog$ by the co-facts in $\cofacts$, and resolution of the original goal succeeds if both such standard SLD resolution of the selected atom and resolution of the remaining goal succeed. 

Finally, rule \rn{main} defines the main judgment $\cosemMain{\progcoax{\prog}{\cofacts}} {G} {E}$ corresponding to an initial step of resolution, where the goal is syntactic, that is, no equation has been produced yet.  

In the following, to avoid confusion with proof trees considered in \refToSection{cofacts}, trees obtained \FD{by} instantiating the rules of big-step operational semantics are called \emph{derivations}.

\paragraph{A derivation example}

{Assuming $\prog$ and $\cofacts$ from the previous section, which define, among others, the predicate \lstinline{max} computing the greatest element of a list, we show a derivation for} the judgment\\ $\cosem{\emptyset}{\progcoax{\prog}{\cofacts}}{\lstmath{max(L,M)}}{E_0}{E_4}$, with\footnote{We use the abbreviated Prolog syntax for lists in this example.}  $E_0=\{\lstmath{L=[1,2|L]}\}$,  $E_4=\{\lstmath{L=[1,2|L],M2=M,M3=2,M2=2,M1=2,M=2}\}$. For keeping the derivation
simpler, in rules (co-step) and (step) we do not manifest some substitutions, but, instead,
we implicitly apply them in the resulting goals.  Furthermore, we \FD{consider the} predicate \lstinline{_ is max(_,_)} 
\FD{to be} predefined, as it is the case in real Prolog systems; for it we have introduced a special semantic rule (labeled with p, for predefined)
that deals with standard inductive predefined predicates (see \refToSection{impl}). Finally, for space reasons we have abbreviated the
names of the applied rules.
$$
\begin{small}
\begin{array}{l}
\Rule{c-s}
{
          \Rule{c-s}
          {
            \Rule{c-h}{
            \Rule{s}        
                {\blindfrac{\der_1}\quad
                  \Rule{e}{}{\sem{\Extended{\prog}{\cofacts}}{\emptygoal}{E_2}{E_2}}{}
                }
                {\sem{\Extended{\prog}{\cofacts}}{\lstmath{max(L,M2)}}{E_0\cup\{\lstmath{M2=M}\}}{E_2}}
                {}
            \qquad
            \blindfrac{\der_2}
}
{\cosem{S_2}{\progcoax{\prog}{\cofacts}}{\lstmath{max(L,M2),M1 is max(2,M2)}}{E_0}{E_3}}{}
            \quad
            \blindfrac{\der_3}
          }
          {\cosem{S_1}{\progcoax{\prog}{\cofacts}}{\lstmath{max([2|L],M1),M is max(1,M1)}}{E_0}{E_4}}
          {}
          {}
          \quad
          \Rule{c-e}
          {}
          {\cosem{\emptyset}{\progcoax{\prog}{\cofacts}}{\emptygoal}{E_4}{E_4}}
          {}
}
{\cosem{\emptyset}{\progcoax{\prog}{\cofacts}}{\lstmath{max(L,M)}}{E_0}{E_4}}
{}
\\
\end{array}
\end{small}
$$
with $E_2=\{\lstmath{L=[1,2|L],M2=M,M3=2,M2=2}\}$, $E_3=\{\lstmath{L=[1,2|L],M2=M,M3=2,M2=2,M1=2}\}$, 
\FD{$S_1=\{\lstmath{max(L,M)}\}$, $S_2 = \{\lstmath{max(L,M), max([2|L],M1)}\}$,  }
and where $\der_1$, $\der_2$, and $\der_3$  are the following derivations:

\begin{flushleft}
$
\begin{small}
\begin{array}{l}
\der_1=
\Rule{s*}
{
  \Rule{e}{}{\sem{\Extended{\prog}{\cofacts}}{\emptygoal}{E_1}{E_1}}{}
  \quad
  \Rule{s}
       {
         \Rule{e}{}{\sem{\Extended{\prog}{\cofacts}}{\emptygoal}{E_2}{E_2}}{}
         \quad
         \Rule{e}{}{\sem{\Extended{\prog}{\cofacts}}{\emptygoal}{E_2}{E_2}}{}
       }
       {\sem{\Extended{\prog}{\cofacts}}{\lstmath{M2 is max(1,M3)}}{E_1}{E_2}}{}
}
{\sem{\Extended{\prog}{\cofacts}}{\lstmath{max([2|L],M3),M2 is max(1,M3)}}{E_0\cup\{\lstmath{M2=M}\}}{E_2}}
{}
\\[4ex]
\der_2 =         \Rule{p}
                {
                  \Rule{e}{}{\sem{\prog}{\emptygoal}{E_3}{E_3}}{}\quad
                  \Rule{c-e}{}{\cosem{S_2}{\progcoax{\prog}{\cofacts}}{\emptygoal}{E_3}{E_3}}{}
                }
                {\cosem{S_2}{\progcoax{\prog}{\cofacts}}{\lstmath{M1 is max(2,M2)}}{E_2}{E_3}}
                {}
\\[4ex]
\der_3=
\Rule{p}
{
  \Rule{e}{}{\sem{\prog}{\emptygoal}{E_4}{E_4}}{}
  \quad
  \Rule{c-e}{}{\cosem{S_1}{\progcoax{\prog}{\cofacts}}{\emptygoal}{E_4}{E_4}}{}
}
{\cosem{S_1}{\progcoax{\prog}{\cofacts}}{\lstmath{M is max(1,M1)}}{E_3}{E_4}}{}
\end{array}
\end{small}
$
\end{flushleft}
with $E_1=\{\lstmath{L=[1,2|L],M2=M,M3=2}\}$.

In the whole derivation the co-fact for \lstinline{max} is used just once in rule (step) marked with (s*) in $\der_1$ derivation;
the co-fact could be employed also in rule (step) in derivation $\der_1$, but without success, since the substitution 
$\{\lstmath{M2=M,M2=1}\}$ cannot satisfy the goal \lstinline{M1 is max(2,M2),M is max(1,M1)}.

\paragraph{Soundness proof}
The big-step operational semantics is sound with respect to the declarative semantics defined in \refToSection{cofacts}. 
That is, if resolution of the goal $G$ succeeds producing a set of equations $E$ which describe a solution of the goal, then (any ground instance of) this solution is a set of atoms which are true in the declarative semantics.
{Formally, set $\Gsol{E}=\{\theta\mid\theta\ \mbox{ground},\ \ApplySubst{t}{\theta} = \ApplySubst{t'}{\theta}\ \mbox{for all}\ t=t'\in E\}$ the set of the \emph{ground solutions} of $E$, necessarily defined on the set $\Var{E}$ of the variables occurring in $E$. In the following, when we pick a substitution from $\Gsol{E}$,  we will implicitly assume that this substitution has $\Var{E}$ as domain.}
Note that the set of ground solutions is antitone with respect to set inclusion, that is, if $E_1 \subseteq E_2$ then $\Gsol{E_2}\subseteq \Gsol{E_1}$.
Soundness can be stated as follows.

\begin{theorem}[Soundness] \label{theo:soundness}
If $\cosemMain{\progcoax{\prog}{\cofacts}}{G}{E}$ holds  then, for each $\theta \in \Gsol{E}$, $\ApplySubst{G}{\theta} \subseteq \Generated{\prog}{\cofacts}$.
\end{theorem}

The proof of soundness is an application of the bounded coinduction principle (\refToTheorem{bcoind}) introduced in \refToSection{cofacts}. For sake of clarity, we provide the proof in a top-down manner, that is, we first give the proof schema, and  then state and prove the three needed subtheorems. 

Set $\OpSem{\prog}{\cofacts}$ the set of atoms which are true in the operational semantics, that is, $\OpSem{\prog}{\cofacts}=\{\ApplySubst{A}{\theta} \mid\ \cosemMain{\progcoax{\prog}{\cofacts}}{A}{E}, \theta \in \Gsol{E}\}$.
We write $\Agree{\theta}{\theta'}$ if $\theta$ and $\theta'$ agree on the common domain, that is, for each $X \in \dom(\theta) {\cap} \dom(\theta')$, $\ApplySubst{X}{\theta} = \ApplySubst{X}{\theta'}$.
Note that, if $\Agree{\theta}{\theta'}$, then $\theta\cup\theta'$ is a well-defined substitution; analogously, we write $\Agree{E_1}{E_2}$ if each ground solution of $E_1$ agrees with a ground solution of $E_2$ on the common domain, that is, for each $\theta \in \Gsol{E_1}$  there exists $\theta' \in \Gsol{E_2}$ such that $\Agree{\theta}{\theta'}$.

\begin{proof}[Proof of \refToTheorem{soundness}]
{Thanks to \refToLemma{cosem-properties}(4), the statement can be equivalently formulated as\\
$\OpSem{\prog}{\cofacts} \subseteq \Generated{\prog}{\cofacts}$.}
By bounded coinduction, we have to show that:
\begin{description}
\item[Boundedness] $\OpSem{\prog}{\cofacts} \subseteq \Ind{\Extended{\prog}{\cofacts}}$, that is, each $A\in\OpSem{\prog}{\cofacts}$ has a finite proof tree in ${\Extended{\prog}{\cofacts}}$.
\item[Consistency] $\OpSem{\prog}{\cofacts} \subseteq \Op{\prog}(\OpSem{\prog}{\cofacts})$, that is, \FD{for all $\ApplySubst{A}{\theta} \in \OpSem{\prog}{\cofacts}$,\\
 there are $\ApplySubst{B_1}{\theta}, \ldots, \ApplySubst{B_n}{\theta} \in\OpSem{\prog}{\cofacts}$ such that  $\clause{\ApplySubst{A}{\theta}}{\ApplySubst{B_1}{\theta}, \ldots, \ApplySubst{B_n}{\theta}}  \in\coGround{\prog}$. }
\end{description}

The two conditions can be proved as follows. 
\begin{itemize}
\item To prove boundedness, we have to show that, if $\cosemMain{\progcoax{\prog}{\cofacts}}{A}{E}, \theta \in \Gsol{E}$, then $\ApplySubst{A}{\theta}\in\Ind{\Extended{\prog}{\cofacts}}$. This can be proved by two steps:
\begin{itemize}
\item $\cosemMain{\progcoax{\prog}{\cofacts}}{A}{E}$ implies \FD{$\sem{\Extended{\prog}{\cofacts}}{A}{\emptyset}{E'}$ with $\theta \in \Gsol{E'}$, that follows from \refToTheorem{coherency}, since $E' \subseteq E$ and so $\Gsol{E} \subseteq \Gsol{E'}$. }
\item  $\sem{\Extended{\prog}{\cofacts}}{A}{\emptyset}{E}$ implies $\ApplySubst{A}{\theta} \in \Ind{\Extended{\prog}{\cofacts}}$ (by \refToTheorem{ind-soundness}) 
\end{itemize}
\item Consistency can be proved as follows:
\begin{itemize}
\item $\cosemMain{\progcoax{\prog}{\cofacts}}{A}{E}$ implies that \FD{there exist $B_1, \ldots, B_n$ such that  $\clause{\ApplySubst{A}{\theta}}{\ApplySubst{B_1}{\theta}, \ldots, \ApplySubst{B_n}{\theta}} \in \coGround{\prog}$ and} $\cosemMain{\progcoax{\prog}{\cofacts}}{B_1, \ldots, B_n}{E'}$ with $\Agree{E}{E'}$ (by \refToTheorem{consistency})
\item $\cosemMain{\progcoax{\prog}{\cofacts}}{B_1, \ldots, B_n}{E'}$ implies $\cosemMain{\progcoax{\prog}{\cofacts}}{B_i}{E'_i}$ with $E'_i \subseteq E'$ (by \refToLemma{cosem-properties}(4)), \FD{that implies $\Gsol{E'} \subseteq \Gsol{E'_i}$. }
\item $\cosemMain{\progcoax{\prog}{\cofacts}}{B_i}{E'_i}$ implies $\ApplySubst{B_i}{\theta} \in \OpSem{\prog}{\cofacts}$,  \FD{that follows from \EZ{the following facts:}
\begin{itemize}
\item there exists $\sigma \in \Gsol{E'}$ such that $\Agree{\theta}{\sigma}$, since $\Agree{E}{E'}$
\item $\sigma \in \Gsol{E'_i}$, hence $\Var{B_i} \subseteq \dom(\sigma)$
\item $\ApplySubst{B_i}{\sigma} \in \OpSem{\prog}{\cofacts}$  by definition of $\OpSem{\prog}{\cofacts}$
\item $\ApplySubst{B_i}{\theta} = \ApplySubst{B_i}{\sigma}$ since $\Agree{\theta}{\sigma}$ and $\Var{B_i} \subseteq \dom(\theta) \cap \dom(\sigma)$.
\end{itemize} }
\end{itemize} 
\end{itemize}
\end{proof}

Now we detail the proof, by stating and proving the three theorems \FD{and the lemma} used above. Some proofs are given in \refToApp{proof}.
We first introduce the following notations. 
If $A=p(t_1, \ldots, t_n)$, and $B=p(s_1, \ldots, s_n)$, then we write $\UnifyEq{A}{B}$ for the set of equations that represents the unification of $A$ with $B$, that is, $\UnifyEq{A}{B} = \{s_1=t_1, \ldots, s_n=t_n\}$.
\EZComm{non ho capito questo: Note that, given the set of equations $E=\{s=t, t=r\}$, we have $\Gsol{E} = \Gsol{E \cup \{s=r\}}$, that can be trivially extended to more than two equations.}




The boundedness condition is obtained by \refToTheorem{coherency} and \refToTheorem{ind-soundness}.
The former states that resolution under coinductive hypotheses implies standard SLD resolution in the program enriched by co-facts, the {latter} states that standard SLD resolution is sound with respect to the inductive semantics. 

We start with an auxiliary lemma.

 \begin{lemma} \label{lemma:sem-properties}
If $\sem{\prog}{G}{E_1}{E_2}$ holds, then 
\begin{enumerate}
\item if  $E'_1 \subseteq E_1$ then $\sem{\prog}{G}{E'_1}{E'_2}$ holds, with $E'_2 \subseteq E_2$ 
\item if $E \cup E_2$ is solvable then $\sem{\prog}{G}{E_1\cup E}{E_2 \cup E}$ holds
\end{enumerate}
\end{lemma}
\begin{proof}
{Straightforward} induction on inference rules.
\end{proof}

\begin{corollary} \label{cor:ind-eq-substitution}
If $\sem{\prog}{G}{E_1 \cup E}{E_2}$ holds and $E_2 \cup E'$ is solvable then $\sem{\prog}{G}{E_1\cup E'}{E''}$ holds,  with $E' \subseteq E'' \subseteq E_2 \cup E'$.
\end{corollary}

The following theorem states that if resolution under coinductive hypotheses of a goal succeeds, then standard SLD resolution of the goal in the program enriched by co-facts succeeds as well.\EZComm{non vale il viceversa, corrisponde al fatto che la semantica induttiva del programma esteso \`e il bound; potrebbe sembrare un po' strano che $S$ sparisca.}

\begin{theorem}\label{theo:coherency}
If $\cosem{S}{\progcoax{\prog}{\cofacts}}{G}{E_1}{E_2}$ holds, then $\sem{\Extended{\prog}{\cofacts}}{G}{E_1}{E'_2}$ holds with $E'_2 \subseteq E_2$.
\end{theorem}
\begin{proof}
By induction on the inference rules which define $\cosem{S}{\progcoax{\prog}{\cofacts}}{G}{E_1}{E_2}$.
\begin{description}
\item[\rn{co-empty}] We trivially conclude by rule \rn{empty}.
\item[\rn{co-step}] By inductive hypothesis we get that $\sem{\Extended{\prog}{\cofacts}}{A_1, \ldots, A_n}{E_1\cup \UnifyEq{A}{A'}}{E'_2}$ and $\sem{\Extended{\prog}{\cofacts}}{G_1, G_2}{E_2}{E'_3}$ hold with $E'_2 \subseteq E_2$, $E'_3 \subseteq E_3$ and $\clause{A'}{A_1, \ldots, A_n}$ a {fresh} renaming of a clause in $\Extended{\prog}{\cofacts}$. 
By \refToLemma{sem-properties}(1) we get that $\sem{\Extended{\prog}{\cofacts}}{G_1, G_2}{E'_2}{E''_3}$ holds, with $E''_3\subseteq E'_3 \subseteq E_3$, thus by rule \rn{step} we get the thesis.
\item[\rn{co-hyp}]  By hypothesis and by \refToLemma{sem-properties}(1) we get that $\sem{\Extended{\prog}{\cofacts}}{A}{E_1}{E'_2}$ holds, with $E'_2 \subseteq E_2$. 
From this, by rule \rn{step}, it follows that $\sem{\Extended{\prog}{\cofacts}}{A_1, \ldots, A_n}{E_1 \cup \UnifyEq{A}{A'}}{E'_2}$ holds, with $\clause{A'}{A_1, \ldots, A_n}$ a {fresh} renaming of a clause in $\Extended{\prog}{\cofacts}$. 
By inductive hypothesis we get $\sem{\Extended{\prog}{\cofacts}}{G_1, G_2}{E_2}{E'_3}$ holds with $E'_3 \subseteq E_3$ and by \refToLemma{sem-properties}(1) we get that $\sem{\Extended{\prog}{\cofacts}}{G_1, G_2}{E'_2}{E''_3}$ holds with $E''_3 \subseteq E'_3 \subseteq E_3$. 
Therefore by rule \rn{step} we get the thesis.
\end{description}
\end{proof}

The following theorem states that if standard SLD resolution of a goal succeeds, producing a set of equations $E$ which describes a solution of the goal, then this solution is a set of atoms which are true in the inductive semantics. In other words, this theorem states the standard soundness property of SLD resolution. The theorem could be derived by showing equivalence of the $\sem{\prog}{A}{\emptyset}{E}$ judgment with the traditional small-step definition of SLD resolution (as done in \cite{AnconaDovier15} for co-SLD resolution), and relying on the well-know soundness of the latter. A direct proof can be done by induction on the rules which define $\sem{\prog}{A}{\emptyset}{E}$, as stated below 
\EZComm{probabilmente in realt\`a va generalizzata la tesi?}

\begin{theorem}[Soundness of standard SLD resolution] \label{theo:ind-soundness}
If $\sem{\prog}{A}{\emptyset}{E}$ then, for each $\theta \in \Gsol{E}$, $\ApplySubst{A}{\theta} \in \Ind{\prog}$.
\end{theorem}


We state now some lemmas needed to prove \refToTheorem{consistency}.

The following lemma states some properties of the judgment $\cosem{\cohyp}{\progcoax{\prog}{\cofacts}}{G}{E}{E'}$.
\FD{In particular, points  1 and 2 state \EZ{a form of monotonicity} with respect to the set of coinductive \EZ{hypotheses} and the set of input equations: \EZ{that is,} we can freely add coinductive \EZ{hypotheses} and remove input equations preserving the results of the evaluation.
Point 3 states another monotonicity property: we can add input equations \EZ{provided that this addition does not} break the solvability of the output system of equations.
Finally, point 4 states a decomposition property: the evaluation of a single atom in a goal  produces a subset of the equations produced by the entire goal.}
\begin{lemma} \label{lemma:cosem-properties}
If $\cosem{S}{\progcoax{\prog}{\cofacts}}{G}{E_1}{E_2}$ holds, then:
\begin{enumerate}
\item  \FD{if $S \subseteq S'$} then $\cosem{S'}{\progcoax{\prog}{\cofacts}}{G}{E_1}{E_2}$ holds  
\item if $E'_1 \subseteq E_1$  then $\cosem{S}{\progcoax{\prog}{\cofacts}}{G}{E'_1}{E'_2}$ holds \FD{for some} $E'_2 \subseteq E_2$
\item if $E\cup E_2$ is solvable then $\cosem{S}{\progcoax{\prog}{\cofacts}}{G}{E_1 \cup E}{E_2 \cup E}$ holds
\item if $G = G_1, A, G_2$ then $\cosem{S}{\progcoax{\prog}{\cofacts}}{A}{E_1}{E}$ holds \FD{for some}  $E \subseteq E_2$
\end{enumerate}
\end{lemma}
\begin{proof}
{Straightforward} induction on inference rules.
\end{proof}

As a consequence we get the following result.
\begin{corollary} \label{cor:eq-substitution}
If $\cosem{S}{\progcoax{\prog}{\cofacts}}{G}{E_1 \cup E}{E_2}$ holds and $E_2 \cup E'$ is solvable, then $\cosem{S}{\progcoax{\prog}{\cofacts}}{G}{E_1 \cup E'}{E''}$ holds, with $E' \subseteq E'' \subseteq E_2 \cup E'$.
\end{corollary} 
\begin{proof}
From point 2 and 3 of \refToLemma{cosem-properties}.
\end{proof}

\begin{lemma} \label{lemma:atom-substitution}
If $\cosem{S \cup \{A\}}{\progcoax{\prog}{\cofacts}}{G}{E_1}{E_2}$ holds, and $E_2 \cup \UnifyEq{A}{A'}$ is solvable, then 
$\cosem{S \cup \{A'\}}{\progcoax{\prog}{\cofacts}}{G}{E_1}{E}$ holds and 
$\Gsol{E_2 \cup \UnifyEq{A}{A'}} \subseteq \Gsol{E}$.
\end{lemma}
The proof can be found in the Appendix.

\begin{lemma} \label{lemma:join-equations}
{If $\sigma_1 \in \Gsol{E_1}$, $\sigma_2 \in \Gsol{E_2}$, and $\Agree{\sigma_1}{\sigma_2}$, then $\sigma_1 \cup \sigma_2 \in \Gsol{E_1 \cup E_2}$.}
\end{lemma}
\begin{proof}
Trivial.
\end{proof}

\begin{theorem}[Consistency]\label{theo:consistency} 
If $\cosemMain{\progcoax{\prog}{\cofacts}}{A}{E}$ holds and $\theta \in \Gsol{E}$, then \FD{there exist atoms $B_1, \ldots, B_n$ such that $\clause{\ApplySubst{A}{\theta}}{\ApplySubst{B_1}{\theta}, \ldots, \ApplySubst{B_n}{\theta}} \in \coGround{\prog}$  and }$\cosemMain{\progcoax{\prog}{\cofacts}}{B_1,\ldots,B_n}{E'}$ holds with $\Agree{E}{E'}$.
\end{theorem}
\EZComm{questa prova probabilmente sar\`a da mettere in appendice}
\begin{proof}
Since by hypothesis $\cosemMain{\progcoax{\prog}{\cofacts}}{A}{E}$ holds, we have necessarily applied rule \rn{main}, hence
$\cosem{\emptyset}{\progcoax{\prog}{\cofacts}}{A}{\emptyset}{E}$ holds. \FD{We have also necessarily applied} rule \rn{co-step}, hence $\cosem{\{A\}}{\progcoax{\prog}{\cofacts}}{B_1, \ldots, B_n}{\UnifyEq{A}{A'}}{E}$ holds, 
for some {fresh} renamed clause $\clause{A'}{B_1, \ldots, B_n}$ in $\prog$ such that $\UnifyEq{A}{A'}$ is solvable. 
\FD{Therefore for all $\theta \in \Gsol{E}$ we have $\clause{\ApplySubst{A'}{\theta}}{\ApplySubst{B_1}{\theta}, \ldots, \ApplySubst{B_n}{\theta}} \in \coGround{\prog}$, and since $\UnifyEq{A}{A'} \subseteq E$, that implies $\ApplySubst{A}{\theta} = \ApplySubst{A'}{\theta}$, we get $\clause{\ApplySubst{A}{\theta}}{\ApplySubst{B_1}{\theta}, \ldots, \ApplySubst{B_n}{\theta}} \in \coGround{\prog}$. }\\
Now we have to show that $\cosem{\emptyset}{\progcoax{\prog}{\cofacts}}{B_1, \ldots, B_n}{\emptyset}{E'}$ holds for some $E'$ such that $\Agree{E}{E'}$.
We prove this statement in two steps.
\begin{enumerate}
\item First we prove that $\cosem{\emptyset}{\progcoax{\prog}{\cofacts}}{B_1, \ldots, B_n}{\UnifyEq{A}{A'}}{E''}$ for some $E''$ with $\Agree{E}{E''}$ 
\item Then we remove the equations in $\UnifyEq{A}{A'}$.
\end{enumerate}
The first part can be proved \FD{by showing that for each judgement $\cosem{S\cup\{A\}}{\progcoax{\prog}{\cofacts}}{G}{E_1}{E_2}$ in the derivation of $\cosem{\{A\}}{\progcoax{\prog}{\cofacts}}{B_1, \ldots, B_n}{\UnifyEq{A}{A'}}{E}$, also the judgement $\cosem{S}{\progcoax{\prog}{\cofacts}}{G}{E_1}{E'_2}$ holds for some $E'_2$ such that $\Agree{E}{E'_2}$. 
This statement can be proved by induction on the derivation as follows. } 
\begin{description}
\item [\rn{co-empty}] Trivial.
\item [\rn{co-step}] The rule has the following shape:
\begin{footnotesize}
\[
\Rule{co-step}
{
\begin{array}{c}
	\cosem{S \cup \{A\} \cup \{B\}}{\progcoax{\prog}{\cofacts}}{C_1, \ldots, C_m}{E_1 \cup \UnifyEq{B}{C}}{E_2}\\
	\cosem{S\cup\{A\}}{\progcoax{\prog}{\cofacts}}{G_1, G_2}{E_2}{E_3}
\end{array}
}
{ \cosem{S \cup \{A\}}{\progcoax{\prog}{\cofacts}}{G_1, B, G_2}{E_1}{E_3} }
{ \begin{array}{l}
	\clause{C}{C_1, \ldots, C_m} \mbox{ in } \prog\\
	\Var{C, C_1, \ldots, C_m} \mbox{ {fresh}}\\
	E_1 \cup \UnifyEq{B}{C} \mbox{ solvable}
\end{array} }
\]
\end{footnotesize}
By inductive hypothesis we have $\cosem{S\cup \{B\}}{\progcoax{\prog}{\cofacts}}{C_1, \ldots, C_m}{E_1\cup \UnifyEq{B}{C}}{E'_2}$ and $\cosem{S}{\progcoax{\prog}{\cofacts}}{G_1, G_2}{E_2}{E'_3}$, with $\Agree{E}{E'_2}$ and $\Agree{E}{E'_3}$.
Consider $\theta \in \Gsol{E}$, since $E_2 \subseteq E_3 \subseteq E$, we get $\Gsol{E} \subseteq \Gsol{E_3} \subseteq \Gsol{E_2}$, therefore $\theta \in \Gsol{E_2}$ that is, $\Var{E_2} \subseteq \dom(\theta)$;
moreover there exist $\sigma_1 \in \Gsol{E'_2}$ and $\sigma_2 \in \Gsol{E'_3}$ such that $\Agree{\theta}{\sigma_1}$ and $\Agree{\theta}{\sigma_2}$, because $\Agree{E}{E'_2}$ and $\Agree{E}{E'_3}$.
Since all introduced variables are fresh, we have that $\Var{E'_2} \cap \Var{E'_3} \subseteq \Var{E}$, that is, $\dom(\sigma_1) \cap \dom(\sigma_2) \subseteq \Var{E} = \dom(\theta)$, 
thus, from what we know above we get that for each $X \in \dom(\sigma_1) \cap \dom(\sigma_2)$, $\ApplySubst{X}{\sigma_1} = \ApplySubst{X}{\theta}=\ApplySubst{X}{\sigma_2}$, that is, $\Agree{\sigma_1}{\sigma_2}$ and $\Agree{\theta}{\sigma_1 \cup \sigma_2}$.
Therefore by \refToLemma{join-equations} we get that $\sigma_1 \cup \sigma_2 \in \Gsol{E'_2 \cup E'_3}$, that is, $E'_2\cup E'_3$ is solvable and $\Agree{E}{E'_2 \cup E'_3}$.\\
Therefore by \refToCorollary{eq-substitution} we get $\cosem{S}{\progcoax{\prog}{\cofacts}}{G_1, G_2}{E'_2}{E'}$ with $\Gsol{E'_2\cup E'_3} \subseteq \Gsol{E'}$, thus $\Agree{E}{E'}$ and by rule \rn{co-step} we get the thesis.  

\item [\rn{co-hyp}] We consider the case where the rule has the following shape, if $B$ unifies with an atom in $S$ the thesis follows from inductive hypothesis by applying rule \rn{co-hyp}.
\[
\Rule{co-hyp}
{
	\sem{\Extended{\prog}{\cofacts}}{B}{E_1 \cup \UnifyEq{A}{B}}{E_2}
	\Space
	\cosem{S\cup\{A\}}{\progcoax{\prog}{\cofacts}}{G_1, G_2}{E_2}{E_3}
}
{ \cosem{S \cup \{A\}}{\progcoax{\prog}{\cofacts}}{G_1, B, G_2}{E_1}{E_3} }
{ E_1 \cup \UnifyEq{A}{B} \mbox{ solvable} }
\] 
By hypothesis $\cosem{\{A\}}{\progcoax{\prog}{\cofacts}}{B_1, \ldots, B_n}{\UnifyEq{A}{A'}}{E}$ holds with $\clause{A'}{B_1, \ldots, B_n}$ a fresh renaming of a clause in $\prog$. 
Let $\rho$ be a renaming of variables in $\Var{E}$   which maps variables in $\Var{A}$ in themselves, and other variables to fresh variables,
then $\cosem{\{A\}}{\progcoax{\prog}{\cofacts}}{B'_1, \ldots, B'_n}{\UnifyEq{A}{\ApplySubst{A'}{\rho}}}{\ApplySubst{E}{\rho}}$ holds with $B'_i = \ApplySubst{B_i}{\rho}$ 
and $\theta \in \Gsol{E}$ iff $\CompSubst{\Inv{\rho}}{\theta} \in \Gsol{\ApplySubst{E}{\rho}}$ and $\Agree{\theta}{\CompSubst{\Inv{\rho}}{\theta}}$.
Now we have to show that $\ApplySubst{E}{\rho} \cup E_1 \cup \UnifyEq{B}{\ApplySubst{A'}{\rho}}$ is solvable; 
to this aim we consider a substitution $\theta \in \Gsol{E}$, 
and note that, since $E_1 \subseteq E$, $\Gsol{E} \subseteq \Gsol{E_1}$, \FD{this implies }$\theta \in \Gsol{E_1}$,
therefore by \refToLemma{join-equations}, $\theta \cup \CompSubst{\Inv{\rho}}{\theta} \in \Gsol{\ApplySubst{E}{\rho} \cup E_1}$. 
Moreover, since $\UnifyEq{B}{A} \subseteq E$ and $\UnifyEq{A}{\ApplySubst{A'}{\rho}} \subseteq \ApplySubst{E}{\rho}$, we have 
$\ApplySubst{B}{\theta} = \ApplySubst{A}{\theta} = \ApplySubst{A}{(\CompSubst{\Inv{\rho}}{\theta})} = \ApplySubst{\ApplySubst{A'}{\rho}}{(\CompSubst{\Inv{\rho}}{\theta})}$;
therefore, setting  $\sigma = \theta \cup \CompSubst{\Inv{\rho}}{\theta}$ and $E' = \ApplySubst{E}{\rho} \cup E_1 \cup \UnifyEq{B}{\ApplySubst{A'}{\rho}}$, we get that $\sigma  \in \Gsol{E'}$.
Thus by \refToLemma{cosem-properties}(1), \refToCorollary{eq-substitution} and \refToLemma{atom-substitution} we get  $\cosem{S \cup \{B\}}{\progcoax{\prog}{\cofacts}}{B'_1, \ldots, B'_n}{E_1\cup \UnifyEq{B}{\ApplySubst{A'}{\rho}}}{E'_2}$ holds,
with  $\Gsol{E'} \subseteq \Gsol{E'_2}$. \\
By inductive hypothesis we get $\cosem{S}{\progcoax{\prog}{\cofacts}}{G_1, G_2}{E_2}{E'_3}$ with $\Agree{E}{E'_3}$.
\FD{If $\theta \in \Gsol{E}$ there exist $\sigma_1 \in \Gsol{E'_2}$ and $\sigma_2 \in \Gsol{E'_3}$ such that $\Agree{\theta}{\sigma_1}$ and $\Agree{\theta}{\sigma_2}$ 
and, since all introduced variables are fresh,  $\dom(\sigma_1) \cap \dom(\sigma_2) \subseteq \Var{E} = \dom(\theta)$. \EZ{This implies that,} for all $X \in \dom(\sigma_1) \cap \dom(\sigma_2)$, $\ApplySubst{X}{\sigma_1} = \ApplySubst{X}{\theta} = \ApplySubst{X}{\sigma_2}$, that is, $\Agree{\sigma_1}{\sigma_2}$.
Therefore by \refToLemma{join-equations} we get $\sigma_1 \cup \sigma_2 \in \Gsol{E'_2 \cup E'_3}$ and $\Agree{\theta}{\sigma_1 \cup \sigma_2}$. }\\
Finally, by \refToCorollary{eq-substitution} we get that $\cosem{S}{\progcoax{\prog}{\cofacts}}{G_1, G_2}{E'_2}{E''}$ with $\Gsol{E'_2 \cup E'_3} \subseteq \Gsol{E''}$, so the thesis follows by application of rule \rn{co-step}.
\end{description}

At this point we have proved that $\cosem{\emptyset}{\progcoax{\prog}{\cofacts}}{B_1, \ldots, B_n}{\UnifyEq{A}{A'}}{E'}$ holds with $\Agree{E}{E'}$.
Now applying \refToLemma{cosem-properties}(2) we get that $\cosem{\emptyset}{\progcoax{\prog}{\cofacts}}{B_1, \ldots, B_n}{\emptyset}{E''}$ with \FD{$E'' \subseteq E'$. Thus $\Gsol{E'} \subseteq \Gsol{E''}$, hence $\Agree{E}{E''}$. }
Therefore by rule \rn{main} we get $\cosemMain{\progcoax{\prog}{\cofacts}}{B_1, \ldots, B_n}{E''}$.
\end{proof}

\section{Implementation}\label{sect:impl}
We have implemented a prototype meta-interpreter in SWI-Prolog, driven by
the rules of the big-step operational semantics defined in \refToSection{big-step};
the complete source code can be found in \refToApp{code}, and it is publicly available
on the Web\footnote{At \url{http://www.disi.unige.it/person/AnconaD/Software}.}
together with a unit \EZ{test-suite}. 

Tests include the predicates on lists considered in the examples provided in \refToSection{cofacts}, 
but also other predicates on lists, as well as predicates defined on repeating decimals, grammars, graphs,
and infinite regular trees \cite{AnconaEtAl17}. 

The three main predicates \lstinline{solve/1}, \lstinline{solve_gfp/2}, and \lstinline{solve_lfp/2} implement
the semantic judgments $\cosemMain{\progcoax{\prog}{\cofacts}}{G}{E}$, $\cosem{S}{\progcoax{\prog}{\cofacts}}{G}{E}{E'}$,
and $\sem{\prog}{G}{E}{E'}$, respectively. As usual, sets of equations do not need to be explicitly manipulated by the meta-interpreter,
which relies on the built-in unification mechanism offered by the Prolog interpreter, therefore the above mentioned predicates 
take less arguments than their corresponding judgments. Despite this fact,  \lstinline{solve_gfp}, and \lstinline{solve_lfp} have the same arity,
because  \lstinline{solve_lfp} needs an extra argument (in comparison to $\sem{\prog}{G}{E}{E'}$) to guarantee termination for some
queries (see more explanations below).

Differently from the operational semantics, the meta-interpreter is fully deterministic and follows the standard selection rules
for goal atoms (leftmost) and clauses (topmost/leftmost, as supported by the system predicate \lstinline{clause}).

Despite the meta-interpreter is driven by the semantic rules, there is no one-to-one correspondence between the rules and the 
meta-interpreter clauses, for several reasons explained below.

The meta-interpreter allows a correct management of predefined predicates which can be identified thanks to the SWI-Prolog system predicates;
for instance, the examples shown in \refToSection{cofacts} uses the built-in predicates \lstinline{>}, \lstinline{\=}, and \lstinline{max}.
To \EZ{handle such predicates}, we have added a specific clause which does not have a counterpart in the operational semantics: if an atom uses a predefined and necessarily inductive (either built-in, or defined in the standard library) predicate, then the meta-interpreter delegates its resolution to the Prolog interpreter. 

Since the meta-programming facilities manage goals as non-empty sequences of atoms (the empty goal is represented by
the single atom \lstinline{true}), there are no clauses corresponding to the semantic rules \rn{empty} and \rn{co-empty}; 
rather, there are two clauses (named \lstinline{seq} and \lstinline{co-seq}) which simply decompose non-singleton goals in their
leftmost atom and rest of atoms, while the clauses dealing with singleton goals correspond to \rn{step}, \rn{co-step}, and \rn{co-hyp} rules, with the difference that in their bodies there is no atom for resolving the remaining goal, as happens for the corresponding judgment with $G_1,G_2$.  

\DA{
From the complexity results concerning coinductive programming \cite{AnconaDovier15} we know 
that determining whether a goal succeeds w.r.t. the coinductive semantics is not even semi-decidable, and, hence, the same
 applies also to our extension with co-facts which includes the standard coinductive semantics.
}

\DA{
The interpreter limits non termination in two ways:
\begin{itemize}
\item a cut is inserted right after atom \lstinline{co_find(Atom, AtomList)} in the body of the clause 
for \lstinline{solve_gfp} corresponding to rule  \rn{co-hyp}; in this way, the clause corresponding
to \rn{co-step} is never applied for an atom which unifies with an element in the list \lstinline{AtomList} of the
coinductive hypotheses (that is, rule \rn{co-hyp} is applicable); this ensures that non termination is avoided when a ``loop''
in the proof tree is detected. Of course, with the insertion of such a cut we may miss some correct answer \cite{Ancona13}.
 \item we have inserted an additional clause (called \lstinline{cut}) for the \lstinline{solve_lfp} predicate which
uses an association to map encountered atoms (modulo unification) to the number of times
they have been already processed by the meta-interpreter; clause \lstinline{step} updates such an
association, by checking whether the currently processed atom unifies with an atom already present in the association;
since \lstinline{solve_lfp} has to build a finite proof tree, no substitution is applied when unification succeeds. 
Clause \lstinline{cut}, which precedes clause \lstinline{step}, fails with no backtracking if, according to the association, the current atom 
has been already processed twice. Also in this case, we ensure termination at the cost of missing some correct answer.
\end{itemize}
}

\DA{
Clause \lstinline{cut} is also responsible for co-facts: it attempts to apply co-facts only when the corresponding atom is found in the association (but only for the first time); in this way termination is guaranteed in more cases.
Since co-facts are managed as facts with the system predicate \lstinline{clause}, it has been pretty easy to 
extend the meta-interpreter to take into account also the possibility of defining \emph{co-clauses} (see further comments in \refToSection{conclu}). 
}

\section{Conclusion}\label{sect:conclu}

We have proposed a generalized logic programming paradigm inspired by the notion of inference system with coaxioms \cite{AnconaEtAl17},
where it is possible to consider interpretations which are between the inductive and the coinductive one.
At the model-theoretic level (declarative semantics) this has been achieved by taking as semantics of a program the largest co-model of the 
program included in the least model of the program enriched by co-facts; 
at the operational level we have defined a big-step operational semantics which is a refinement of co-SLD resolution: 
when the same goal is encountered the second time, its standard SLD resolution is triggered in the program enriched by co-facts.

We have proved that such an operational semantics is sound with respect to the declarative semantics; furthermore, the big-step
semantics rules have driven the implementation of a prototype meta-interpreter that has allowed us to successfully experiment our
generalized logic programming paradigm.

Coinductive logic programming has been initially investigated and implemented by Simon et al. \cite{SimonEtAl06,SimonEtAl07};
since the earlier stages, \EZ{the problem} has been recognized that not all predicate definitions require a coinductive interpretation;
for instance, Simon et al. have pointed out this issue for the \lstinline{member} predicate, whose semantics is inductive even when
infinite lists are considered. Anyway, to our knowledge, no current implementation of coinductive logic programming supports
the extension of inductive predicates for the complete Herbrand base, neither the ability of defining programs 
where  coinductive and inductive predicates can be freely mixed together. 

Similar issues have been investigated in the context of coinductive functional \cite{JeanninEtAl12,JeanninEtAl13},
and object-oriented \cite{AnconaZucca12,AnconaZucca13} paradigms, but the proposed solutions lack proof principles
useful for proving correctness of programs written in these extended paradigms. 

We have commented in \refToSection{impl} that the prototype meta-interpreter naturally supports not only co-facts, but also
co-clauses, and the big-step semantics can be trivially extended to take into account this possibility, but it would be
interesting to investigate whether this generalization has a natural counterpart at the level of the declarative semantics. 

For what concerns the implementation, much more work is required to guarantee that the extension of logic programming
with co-facts can be effectively used in practice.
\DA{The meta-interpreter offers a simple solution to rapid prototyping 
of an implementation of the operational semantics by exploiting
the reflection facilities of Prolog, but is far from being an efficient solution.
Furthermore, the combination of the two predicates \lstinline{solve_gfp} and \lstinline{solve_lfp} for 
coinductive and inductive reasoning, respectively, has a bad impact on the performance because
in fact a proof tree needs to be built twice. It would be interesting to investigate more clever algorithms
to avoid such a duplication, or, at least, identify conditions on logic programs which are sufficient
to guarantee more efficient implementations.}

\DA{
Another direction for further work consists in studying restricted classes of logic programs
for which the operational semantics of co-facts presented here turns out to be sound and complete.
}

\bibliographystyle{eptcs}
\bibliography{main}

\begin{thebibliography}{10}
\providecommand{\bibitemdeclare}[2]{}
\providecommand{\surnamestart}{}
\providecommand{\surnameend}{}
\providecommand{\urlprefix}{Available at }
\providecommand{\url}[1]{\texttt{#1}}
\providecommand{\href}[2]{\texttt{#2}}
\providecommand{\urlalt}[2]{\href{#1}{#2}}
\providecommand{\doi}[1]{doi:\urlalt{http://dx.doi.org/#1}{#1}}
\providecommand{\bibinfo}[2]{#2}

\bibitemdeclare{article}{AdamekEtAl06}
\bibitem{AdamekEtAl06}
\bibinfo{author}{Jir{\'{\i}} \surnamestart Ad{\'{a}}mek\surnameend},
  \bibinfo{author}{Stefan \surnamestart Milius\surnameend} \&
  \bibinfo{author}{Jiri \surnamestart Velebil\surnameend}
  (\bibinfo{year}{2006}): \emph{\bibinfo{title}{Iterative algebras at work}}.
\newblock {\sl \bibinfo{journal}{Mathematical Structures in Computer Science}}
  \bibinfo{volume}{16}(\bibinfo{number}{6}), pp. \bibinfo{pages}{1085--1131},
  \doi{10.1017/S0960129506005706}.

\bibitemdeclare{article}{Ancona13}
\bibitem{Ancona13}
\bibinfo{author}{Davide \surnamestart Ancona\surnameend}
  (\bibinfo{year}{2013}): \emph{\bibinfo{title}{Regular corecursion in
  Prolog}}.
\newblock {\sl \bibinfo{journal}{Computer Languages, Systems {\&} Structures}}
  \bibinfo{volume}{39}(\bibinfo{number}{4}), pp. \bibinfo{pages}{142--162},
  \doi{10.1016/j.cl.2013.05.001}.

\bibitemdeclare{inproceedings}{AnconaEtAl17}
\bibitem{AnconaEtAl17}
\bibinfo{author}{Davide \surnamestart Ancona\surnameend},
  \bibinfo{author}{Francesco \surnamestart Dagnino\surnameend} \&
  \bibinfo{author}{Elena \surnamestart Zucca\surnameend}
  (\bibinfo{year}{2017}): \emph{\bibinfo{title}{Generalizing inference systems
  by coaxioms}}.
\newblock In \bibinfo{editor}{Hongseok \surnamestart Yang\surnameend}, editor:
  {\sl \bibinfo{booktitle}{ESOP 2017 - European Symposium on Programming}},
  {\sl \bibinfo{series}{Lecture Notes in Computer Science}}
  \bibinfo{volume}{10201}, \bibinfo{publisher}{Springer}, pp.
  \bibinfo{pages}{29--55}, \doi{10.1007/978-3-662-54434-1_2}.

\bibitemdeclare{article}{AnconaDovier15}
\bibitem{AnconaDovier15}
\bibinfo{author}{Davide \surnamestart Ancona\surnameend} \&
  \bibinfo{author}{Agostino \surnamestart Dovier\surnameend}
  (\bibinfo{year}{2015}): \emph{\bibinfo{title}{A Theoretical Perspective of
  Coinductive Logic Programming}}.
\newblock {\sl \bibinfo{journal}{Fundamenta Informaticae}}
  \bibinfo{volume}{140}(\bibinfo{number}{3-4}), pp. \bibinfo{pages}{221--246},
  \doi{10.3233/FI-2015-1252}.

\bibitemdeclare{inproceedings}{AnconaZucca12}
\bibitem{AnconaZucca12}
\bibinfo{author}{Davide \surnamestart Ancona\surnameend} \&
  \bibinfo{author}{Elena \surnamestart Zucca\surnameend}
  (\bibinfo{year}{2012}): \emph{\bibinfo{title}{Corecursive {F}eatherweight
  {J}ava}}.
\newblock In \bibinfo{editor}{Wei{-}Ngan \surnamestart Chin\surnameend} \&
  \bibinfo{editor}{Aquinas \surnamestart Hobor\surnameend}, editors: {\sl
  \bibinfo{booktitle}{FTfJP'12 - Formal Techniques for Java-like Programs}},
  \bibinfo{publisher}{ACM Press}, pp. \bibinfo{pages}{3--10},
  \doi{10.1145/2318202.2318205}.

\bibitemdeclare{inproceedings}{AnconaZucca13}
\bibitem{AnconaZucca13}
\bibinfo{author}{Davide \surnamestart Ancona\surnameend} \&
  \bibinfo{author}{Elena \surnamestart Zucca\surnameend}
  (\bibinfo{year}{2013}): \emph{\bibinfo{title}{Safe Corecursion in co{FJ}}}.
\newblock In \bibinfo{editor}{Werner \surnamestart Dietl\surnameend}, editor:
  {\sl \bibinfo{booktitle}{FTfJP'13 - Formal Techniques for Java-like
  Programs}}, \bibinfo{publisher}{ACM Press}, pp. \bibinfo{pages}{2:1--2:7},
  \doi{10.1145/2489804.2489807}.

\bibitemdeclare{book}{Apt97}
\bibitem{Apt97}
\bibinfo{author}{Krzysztof~R. \surnamestart Apt\surnameend}
  (\bibinfo{year}{1997}): \emph{\bibinfo{title}{From logic programming to
  Prolog}}.
\newblock \bibinfo{series}{Prentice Hall International series in computer
  science}, \bibinfo{publisher}{Prentice Hall}.

\bibitemdeclare{techreport}{JeanninEtAl12}
\bibitem{JeanninEtAl12}
\bibinfo{author}{J.~\surnamestart Jeannin\surnameend},
  \bibinfo{author}{D.~\surnamestart Kozen\surnameend} \&
  \bibinfo{author}{A.~\surnamestart Silva\surnameend} (\bibinfo{year}{2012}):
  \emph{\bibinfo{title}{{CoCaml}: Programming with Coinductive Types}}.
\newblock \bibinfo{type}{Technical Report}, \bibinfo{institution}{Computing and
  Information Science, Cornell University}.
\newblock \urlprefix\url{http://hdl.handle.net/1813/30798}.

\bibitemdeclare{inproceedings}{JeanninEtAl13}
\bibitem{JeanninEtAl13}
\bibinfo{author}{J.~\surnamestart Jeannin\surnameend},
  \bibinfo{author}{D.~\surnamestart Kozen\surnameend} \&
  \bibinfo{author}{A.~\surnamestart Silva\surnameend} (\bibinfo{year}{2013}):
  \emph{\bibinfo{title}{Language Constructs for Non-Well-Founded Computation}}.
\newblock In \bibinfo{editor}{Matthias \surnamestart Felleisen\surnameend} \&
  \bibinfo{editor}{Philippa \surnamestart Gardner\surnameend}, editors: {\sl
  \bibinfo{booktitle}{ESOP 2013 - European Symposium on Programming}}, {\sl
  \bibinfo{series}{Lecture Notes in Computer Science}} \bibinfo{volume}{7792},
  \bibinfo{publisher}{Springer}, pp. \bibinfo{pages}{61--80},
  \doi{10.1007/978-3-642-37036-6_4}.

\bibitemdeclare{book}{Lloyd87}
\bibitem{Lloyd87}
\bibinfo{author}{John~W. \surnamestart Lloyd\surnameend}
  (\bibinfo{year}{1987}): \emph{\bibinfo{title}{Foundations of Logic
  Programming, 2nd Edition}}.
\newblock \bibinfo{publisher}{Springer}, \doi{10.1007/978-3-642-83189-8}.

\bibitemdeclare{phdthesis}{Simon06}
\bibitem{Simon06}
\bibinfo{author}{Luke \surnamestart Simon\surnameend} (\bibinfo{year}{2006}):
  \emph{\bibinfo{title}{Extending logic programming with coinduction}}.
\newblock Ph.D. thesis, \bibinfo{school}{University of Texas at Dallas}.

\bibitemdeclare{inproceedings}{SimonEtAl07}
\bibitem{SimonEtAl07}
\bibinfo{author}{Luke \surnamestart Simon\surnameend}, \bibinfo{author}{Ajay
  \surnamestart Bansal\surnameend}, \bibinfo{author}{Ajay \surnamestart
  Mallya\surnameend} \& \bibinfo{author}{Gopal \surnamestart Gupta\surnameend}
  (\bibinfo{year}{2007}): \emph{\bibinfo{title}{Co-Logic Programming: Extending
  Logic Programming with Coinduction}}.
\newblock In \bibinfo{editor}{Lars \surnamestart Arge\surnameend},
  \bibinfo{editor}{Christian \surnamestart Cachin\surnameend},
  \bibinfo{editor}{Tomasz \surnamestart Jurdzinski\surnameend} \&
  \bibinfo{editor}{Andrzej \surnamestart Tarlecki\surnameend}, editors: {\sl
  \bibinfo{booktitle}{ICALP'07 - International Colloquium on Automata,
  Languages and Programming 2003}}, {\sl \bibinfo{series}{Lecture Notes in
  Computer Science}} \bibinfo{volume}{4596}, \bibinfo{publisher}{Springer}, pp.
  \bibinfo{pages}{472--483}, \doi{10.1007/978-3-540-73420-8_42}.

\bibitemdeclare{inproceedings}{SimonEtAl06}
\bibitem{SimonEtAl06}
\bibinfo{author}{Luke \surnamestart Simon\surnameend}, \bibinfo{author}{Ajay
  \surnamestart Mallya\surnameend}, \bibinfo{author}{Ajay \surnamestart
  Bansal\surnameend} \& \bibinfo{author}{Gopal \surnamestart Gupta\surnameend}
  (\bibinfo{year}{2006}): \emph{\bibinfo{title}{Coinductive Logic
  Programming}}.
\newblock In \bibinfo{editor}{Sandro \surnamestart Etalle\surnameend} \&
  \bibinfo{editor}{Miroslaw \surnamestart Truszczynski\surnameend}, editors:
  {\sl \bibinfo{booktitle}{International Conference on Logic Programming}},
  {\sl \bibinfo{series}{Lecture Notes in Computer Science}}
  \bibinfo{volume}{4079}, \bibinfo{publisher}{Springer}, pp.
  \bibinfo{pages}{330--345}, \doi{10.1007/11799573_25}.

\end{thebibliography}
\appendix

\section{Proofs}\label{sect:proof}
\begin{proof}[Proof of \refToLemma{atom-substitution}]
By induction on inference rules.
\begin{description}
\item[\rn{co-empty}] We trivially conclude by rule \rn{empty}.
\item[\rn{co-step}] By hypothesis $E_3 \cup \UnifyEq{A}{A'}$ is solvable, so by inductive hypothesis we get that $\cosem{S \cup \{A'\}}{\progcoax{\prog}{\cofacts}}{G_1, G_2}{E_2}{E'_3}$ holds with $\Gsol{E_3 \cup \UnifyEq{A}{A'}} \subseteq \Gsol{E'_3}$. 
Since $E_2 \subseteq E_3$, $E_2 \cup \UnifyEq{A}{A'}$ is solvable by hypothesis, 
by inductive hypothesis we get $\cosem{S \cup \{A'\} \cup \{B'\}}{\progcoax{\prog}{\cofacts}}{B_1, \ldots B_n}{E_1 \cup \UnifyEq{B}{B'}}{E'_2}$ with $\Gsol{E_2 \cup \UnifyEq{A}{A'}} \subseteq \Gsol{E'_2}$.
Note that $\Gsol{E_3 \cup \UnifyEq{A}{A'}} \subseteq \Gsol{E_2 \cup \UnifyEq{A}{A'}} \subseteq \Gsol{E'_2}$, therefore $\Gsol{E_3 \cup \UnifyEq{A}{A'}} \subseteq \Gsol{E'_2} \cap \Gsol{E'_3}$, that implies $\Gsol{E_3 \cup \UnifyEq{A}{A'}} \subseteq \Gsol{E'_2 \cup E'_3}$.
Therefore by \refToCorollary{eq-substitution} we get that $\cosem{S \cup \{A'\}}{\progcoax{\prog}{\cofacts}}{G_1, G_2}{E'_2}{E''_3}$ holds with $E''_3 \subseteq E'_2 \cup E'_3$, thus $\Gsol{E_3 \cup \UnifyEq{A}{A'}} \subseteq \Gsol{E''_3}$.
Finally we get the thesis by applying rule \rn{co-step}.
\item[\rn{co-hyp}] If the atom $B$ unifies with an atom in $S$  the the thesis follow by the inductive hypothesis applying rule \rn{co-hyp}.\\
Suppose that $B$ unifies with $A$, that is $E_1 \cup \UnifyEq{B}{A}$ is solvable. 
By hypothesis $E_3 \cup \UnifyEq{A}{A'}$ is solvable, thus by inductive hypothesis we get $\cosem{S}{\progcoax{\prog}{\cofacts}}{G_1, G_2}{E_2}{E'_3}$ with $\Gsol{E_3 \cup \UnifyEq{A}{A'}} \subseteq \Gsol{E'_3}$.
Since $E_2 \subseteq E_3$ we know that $\Gsol{E_3\cup \UnifyEq{A}{A'}} \subseteq \Gsol{E_2 \cup \UnifyEq{A}{A'}}$;
moreover since $\UnifyEq{B}{A} \subseteq E_2$, we get that $\Gsol{E_2 \cup \UnifyEq{A}{A'} \cup UnifyEq{B}{A'}} = \Gsol{E_2 \cup \UnifyEq{A}{A'}}$, 
thus, by \refToCorollary{ind-eq-substitution}, we get that $\sem{\Extended{\prog}{\cofacts}}{B}{E_1 \cup \UnifyEq{B}{A'}}{E'_2}$ with $\Gsol{E_2 \cup \UnifyEq{A}{A'}} \subseteq \Gsol{E'_2}$.
Also $\Gsol{E_3 \cup \UnifyEq{A}{A'}} \subseteq \Gsol{E_2 \cup \UnifyEq{A}{A'}} \subseteq \Gsol{E'_2}$, that implies $\Gsol{E_3 \cup \UnifyEq{A}{A'}} \subseteq \Gsol{E'_2} \cap \Gsol{E'_3} \subseteq \Gsol{E'_2 \cup E'_3}$.
Hence by \refToCorollary{eq-substitution} we get that $\cosem{S \cup \{A'\}}{\progcoax{\prog}{\cofacts}}{G_1, G_2}{E'_2}{E''_3}$ holds with $\Gsol{E'_2 \cup E'_3} \subseteq \Gsol{E''_3}$.
Finally we get the thesis by applying rule \rn{co-hyp}.
\end{description}
\end{proof}

\section{Source code}\label{sect:code}

\begin{lstlisting}[basicstyle=\ttfamily\scriptsize,language=Prolog,breaklines=true]
:- module(meta_interpreter, [solve/1]).

:- use_module(library(assoc)). %%% needed to keep track of atom occurrences for finite failure 

solve(Goal) :- solve_gfp([], Goal). %%% (main)

%%% solver for the inductive system with cofacts

solve_lfp(AtomAssoc, (Goal1, Goal2)) :- !, solve_lfp(AtomAssoc, Goal1), solve_lfp(AtomAssoc, Goal2). %%% seq
solve_lfp(_, Atom) :- predefined(Atom), !, Atom. %%% predef
solve_lfp(AtomAssoc, Atom) :- find(Atom, AtomAssoc, Count), (Count<2 -> clause(cofact(Atom), Body), Body; !, fail).  %%% cut
solve_lfp(AtomAssoc, Atom) :- clause(Atom, Body), insert(Atom, AtomAssoc, NewAtomAssoc), solve_lfp(NewAtomAssoc, Body). %%% step

%%% solver for the coinductive system with no cofacts

solve_gfp(AtomList, (Goal1, Goal2)) :- !, solve_gfp(AtomList, Goal1), solve_gfp(AtomList, Goal2). %%% co-seq
solve_gfp(_, Atom) :- predefined(Atom), !, Atom. %%% co-predef
solve_gfp(AtomList, Atom) :- co_find(Atom, AtomList), !, empty_assoc(EmptyAssoc), solve_lfp(EmptyAssoc, Atom). %%% co-hyp
solve_gfp(AtomList, Atom) :- clause(Atom, Body), co_insert(Atom, AtomList, NewAtomList), solve_gfp(NewAtomList, Body). %%% co-step

%%% predefined predicates are interpreted in the standard way

predefined(Atom) :- predicate_property(Atom, built_in), !.
predefined(Atom) :- predicate_property(Atom, file(AbsPath)), file_name_on_path(AbsPath, library(_)), !.

%%% auxiliary predicates for the coinductive solver

co_find(Atom, AtomList) :- member(Atom, AtomList).

co_insert(Atom, AtomList, [Atom|AtomList]).  

%%% auxiliary predicates for the inductive solver

%%% finds if AtomAssoc contains an AtomKey unifiable with Atom, but does not perform unification; if so, returns its corresponding hit count
%%% used by the inductive solver

find(Atom, AtomAssoc, Count) :- retrieve(Atom, AtomAssoc, AtomKey), get_assoc(AtomKey, AtomAssoc, Count).

%%% retrieve the AtomKey in AtomAssoc which is unifiable with Atom; no unification is performed
%%% used by found

retrieve(Atom, AtomAssoc, AtomKey) :- assoc_to_keys(AtomAssoc, AtomKeyList), get(Atom, AtomKeyList, AtomKey).

%%% checks if Atom is unifiable with an atom in AtomList; if so, returns such an atom. No unification is performed
%%% used by retrieve

get(Atom, [UnifiableAtom|_], UnifiableAtom) :- unifiable(Atom, UnifiableAtom, _), !. 
get(Atom, [_|AtomList], UnifiableAtom) :- get(Atom, AtomList, UnifiableAtom). 

%%% insertion in the list of visited atoms when implementing the inductive solver
%%% keeps track of how many times an atom has been hit

insert(Atom, AtomAssoc, NewAtomAssoc) :- 
    retrieve(Atom, AtomAssoc, AtomKey) -> 
	get_assoc(AtomKey, AtomAssoc, Counter), IncCounter is Counter+1, put_assoc(AtomKey, AtomAssoc, IncCounter, NewAtomAssoc); 
        copy_term(Atom, CopiedAtom), put_assoc(CopiedAtom, AtomAssoc, 1, NewAtomAssoc).
\end{lstlisting}

\end{document}